\documentclass[11pt]{article}
\usepackage{paralist}
\usepackage{color}
\usepackage{bm}

\usepackage{soul}

\usepackage[cm]{fullpage}

\usepackage{amsfonts}
\usepackage{amssymb}
\usepackage{graphicx}
\usepackage{algorithm}
\usepackage{algpseudocode}
\usepackage{amsmath}
\usepackage{constants}
\usepackage{amsthm}
\makeatletter
\newtheorem*{rep@theorem}{\rep@title}
\newcommand{\newreptheorem}[2]{%
\newenvironment{rep#1}[1]{%
 \def\rep@title{#2 \ref{##1}}%
 \begin{rep@theorem}}%
 {\end{rep@theorem}}}
\makeatother

\newcount\shortyear\newcount\shorthour\newcount\shortminute
\shorthour=\time\divide\shorthour by 60\shortyear=\shorthour
\multiply\shortyear by 60\shortminute=\time\advance\shortminute by
-\shortyear \shortyear=\year\advance\shortyear by -1900

\def\zeit{\number\shorthour:\ifnum\shortminute<10 0\number\shortminute
\else\number\shortminute\fi}

\usepackage{ifpdf}
\newcommand{\mydriver}{hypertex}
\ifpdf
 \renewcommand{\mydriver}{pdftex}
\fi
\usepackage[breaklinks,\mydriver]{hyperref}

\usepackage[margin=1in]{geometry}

\newcommand{\edgesketch}{\ensuremath{\mathbf{\Delta}}}
\newcommand{\vertexsketch}{\ensuremath{\mathbf{\Lambda}}}
\newcommand{\BHH}{\textrm{BHH}}
\newcommand{\Var}{\textrm{Var}}

\newcommand{\E}{\textrm{E}}
\newcommand{\poly}{\ensuremath{\mathrm{poly}}}

\newcommand{\supp}{\textrm{supp}}

\newcommand{\x}{\textbf{x}}

\newcommand{\y}{\textbf{y}}
\newcommand{\eps}{\varepsilon}

\newcommand{\cc}{\textrm{cc}}
\newcommand{\scc}{\textrm{scc}}
\newcommand{\MST}{\textrm{MST}}

\newcommand{\BB}{\ensuremath{\mathcal{B}}}

\newcommand{\R}{\ensuremath{\mathbb{R}}}

\newcommand{\sketch}{\ensuremath{\mathcal{A}}}
\newcommand{\sparsesketch}{\ensuremath{\mathbf{\Upsilon}}}

\newcommand{\0}{\textbf{0}}

\newcommand{\junk}[1]{{}}

\newtheorem{fact}{Fact}
\newtheorem{claim}{Claim}
\providecommand{\abs}[1]{\left|#1\right|}

\providecommand{\agbracket}[1]{\langle#1\rangle}

\theoremstyle{plain}
\newtheorem{theorem}{Theorem}[section]
\newtheorem{lemma}[theorem]{Lemma}

\newreptheorem{theorem}{Theorem}
\newreptheorem{lemma}{Lemma}
\newtheorem{definition}[theorem]{Definition}

\theoremstyle{definition}

\title{Dynamic Graph Stream Algorithms in $o(n)$ Space}

\author{Zengfeng Huang\footnote{University of New South Wales, Sydney, Australia. Email: \url{zengfeng.huang@unsw.edu.au}.}
\and
Pan Peng\footnote{Department of Computer Science, TU Dortmund; State Key Laboratory of Computer Science, Institute of Software, Chinese Academy of Sciences. Supported by ERC grant No. 307696. Email: \url{pan.peng@tu-dortmund.de}.}
}

\date{}

\begin{document}
	
	
	\begin{titlepage}
		
		\maketitle
		
		\thispagestyle{empty}
		
\begin{abstract}
	In this paper we study graph problems in dynamic streaming model, where the input is defined by a sequence of edge insertions and deletions. As many natural problems require $\Omega(n)$ space, where $n$ is the number of vertices, existing works mainly focused on designing $\tilde{O}(n)$ space algorithms. Although sublinear in the number of edges for dense graphs, it could still be too large for many applications (e.g. $n$ is huge or the graph is sparse).  
	In this work, we give single-pass algorithms beating this space barrier for two classes of problems. 
	
	We present $o(n)$ space algorithms for estimating the number of connected components with additive error $\varepsilon n$ and $(1+\varepsilon)$-approximating the weight of minimum spanning tree, for any small constant $\varepsilon>0$. The latter improves previous $\tilde{O}(n)$ space algorithm given by Ahn et al. (SODA 2012) for connected graphs with bounded edge weights. 
	
	We initiate the study of approximate graph property testing in the dynamic streaming model, where we want to distinguish graphs satisfying the property from graphs that are $\varepsilon$-far from having the property. We consider the problem of testing $k$-edge connectivity, $k$-vertex connectivity, cycle-freeness and bipartiteness (of planar graphs), for which, we provide algorithms using roughly $\tilde{O}(n^{1-\varepsilon})$ space, which is $o(n)$ for any constant $\varepsilon$. 
	
	To complement our algorithms, we present $\Omega(n^{1-O(\varepsilon)})$ space lower bounds for these problems, which show that such a dependence on $\varepsilon$ is necessary. 
\end{abstract}

\end{titlepage}

\section{Introduction}
Graphs or networks are a natural way to describe structural information. For example, users of Facebook and the acquaintance relations among them form a social network, the proteins together with interactions between them define a biological network, and web-pages and hyperlinks give rise to a huge web graph. Due to the rapid development of information technology, many such graphs become extremely large, and are constantly changing, which poses great challenges for analyzing their structures. Over the last decade,  the data stream model~\cite{Mut05:data} has proven to be successful in dealing with \emph{big data}. In this model, the algorithm should make only one pass (or a few passes) over the stream, and use sublinear working space. The time required to output the final answer and process each element is also important. There is a growing body of work studying graph problems over data streams. Graph streams were first considered by Henzinger et al.~\cite{HRR99:stream}, and later have been extensively studied in the \textit{insertion-only} model~(eg., \cite{FKMSZ05:graph,FKM+08:graph,Mut05:data}), where there is no edge deletion in the stream. Recently, starting from the seminal works of Ahn, Guha and McGregor \cite{AGM12:linear, AGM12:sketch}, the interest has shifted to the \emph{dynamic streaming model}, where the edges can be both inserted and deleted (see eg., \cite{KLMS14:spectral,KW14:spanners,ACGMW15:correlation,AKLY15:tight,CCE+15:kernelization,CCH+15:parameter,BS15:sublinear,Kon15:maximum,BHNT15:densest,MTVV15:densest,EHW15:densest,GMT15:vertex}). In this setting, most algorithms designed are linear sketch-based, which is also an effective technique for processing distributed graphs. For more information about graph streaming algorithms see the recent survey by McGregor \cite{Mcg14:survey}. 

For graph streams, both insertion-only and dynamic, the research in the past has mostly focused on the  \emph{semi-streaming model}, in which the algorithms are allowed to use $\tilde{O}(n)$ space, where $n$ is the number vertices in the graph. (For notational convenience, we will use $\tilde{O}(g)$ and $\tilde{\Omega}(g)$ to hide $\poly\log(g)$ factors.) The reason behind this is that even in the insertion-only model, many natural graph problems require $\Omega(n)$ space (e.g. testing if the graph is connected \cite{FKM+08:graph}). 
Note that the allowed space in semi-streaming model is sublinear in the input size as the number of edges of the graph might be as large as $\Omega(n^2)$. However, in many real applications $n$ is huge and the input graph is already very sparse, an $\tilde{O}(n)$ algorithm might be even worse than just storing all the edges. From this perspective, one may naturally ask the question \emph{which kind of problems can be solved with even less space, i.e., $o(n)$ space}.  

To the best of our knowledge, very few results are known in this direction. Chitnis et al. \cite{CCH+15:parameter} and Fafianie and Kratsch~\cite{FK14:parameterized} introduced parameterized graph stream algorithms which may only use $o(n)$ space with some promise of the size of the solution. This parameterized setting has been further investigated in~\cite{CCE+15:kernelization}. In addition, it has been shown that the size of the maximum matching can be approximated within constant factor in $\tilde{O}(n^{4/5})$ space for graphs with bounded arboricity~\cite{EHL+15:matching, CCE+15:kernelization,BS15:sublinear}.


In this paper, we study two classes of graph problems that admit single-pass $o(n)$ space algorithms in the dynamic streaming model. The first class contains the problems of estimating the number of connected components and the weight of minimum spanning tree (MST). We show that one can estimate the number of connected components within an \emph{additive} error of $\varepsilon n$ with $o(n)$ space and post-processing time, for any constant $\varepsilon>0$. We also present an algorithm to $(1+\varepsilon)$-approximate the weight of MST with $o(n)$ space and post-processing time for connected graphs with bounded edge weights, which improves the best known algorithm with $\tilde{O}(n)$ space in the same setting given by Ahn et al.~\cite{AGM12:linear}. It is worthy noting that the problem of estimating the number of connected components within small \emph{multiplicative} error requires $\Omega(n)$ space, as it is generally harder than the problem of (exactly) testing graph connectivity; and that estimating the weight of MST for graphs with arbitrarily large edge weights (e.g., $\Omega(\log n)$) requires $\Omega(n)$ space (see Theorem~\ref{thm:lower_mst}). Previously these two problems have been studied in 
the framework of \emph{sublinear time algorithms}~(see eg.~\cite{CRT05:MST,RS11:sublinear}).

The second class consists of problems that are relaxations of deciding graph properties. Given a huge graph, it is very useful to know whether the graph has some predetermined property, such as $k$-connectivity, bipartiteness, cycle-freeness and etc., which provide valuable information about the graph.
However, besides the requirement of $\Omega(n)$ space, exactly testing of these properties sometimes is a too strong requirement for analyzing highly dynamic graphs, since the answer may change in the next second due to an insertion or deletion of a single edge. In this paper, we initiate the study of \emph{approximate graph property testing} in the dynamic streaming model: we want to test whether a graph satisfies some property or one has to modify a small constant fraction of edges to make it have the property. This notion of approximation is adapted from the framework of \emph{property testing}~\cite{GGR98:testing,GR02:property,PR02:diameter}, and a large number of existing literatures have given efficient testing algorithms (called \emph{testers}) for many properties under different query models (see surveys~\cite{Gol11:introduction,Ron10:algorithmic}). We show that some fundamental properties can be tested in both $o(n)$ space and post-processing time in the dynamic streaming model and we also present close lower bounds for these problems which hold even in the insertion-only model. We remark that McGregor~\cite{Mcg14:survey} also suggested to study the (approximate) property testers in graph streaming model, and asked whether more space-efficient algorithms exist for these problems, and we thus give affirmative answer to this question.

\subsection{Our results} 
Now we formally state our main results. Our results regarding estimating the number of connected components and the MST weight are as follows.
\begin{itemize}
	\item {\bf Estimating the number of connected components.} We present a dynamic streaming algorithm that estimates the number of connected components within additive error $\varepsilon n$ in $\tilde{O}(n^{1-\varepsilon +\varepsilon^{q+1}})$ space and post-processing time for any constant $q\geq 1$. We note that a lower bound of $\Omega(n^{1-O(\varepsilon)})$ for this problem follows from the work~\cite{VY11:cycle}.
	\item {\bf Estimating the weight of minimum spanning tree (MST).}
	In this problem, we want to estimate the weight of the MST of a graph with edge weights in the set $\{1,2,\cdots,W\}$. We give a dynamic streaming algorithm that computes a $(1+\varepsilon)$-approximation of the MST weight and uses space and post-processing time $\tilde{O}(Wn^{1-\frac{\varepsilon}{W-1}+\frac{\varepsilon^t}{(W-1)^t}})$ for any constant $t\geq 1$. By an argument in \cite{CRT05:MST}, the result can be extended to non-integral weights, as long as the ratio between the largest and the smallest weight is bounded.
	A space lower bound of $\Omega(n^{1-\frac{4\varepsilon}{W-1}})$ is shown for this problem.
\end{itemize}

We also present approximate testing algorithms for a number of fundamental graph properties. Before stating the performance of these algorithms, we first introduce some definitions. Given a graph property $\Pi$, an $m$-edge graph $G$ is called \emph{$\varepsilon$-far from having $\Pi$} if one has to modify more than $\varepsilon m$ edges of $G$ to get a graph $G'$ satisfying $\Pi$. This distance definition is adapted from~\cite{PR02:diameter} and is most suitable for general graphs where neither edge density nor maximum degree is restricted. We call an algorithm a \emph{(dynamic) streaming tester} for $\Pi$, if it makes a single-pass over a stream of edge insertions and deletions, with probability at least $2/3$, accepts any graph satisfying $\Pi$, and rejects any graph that is $\varepsilon$-far from having $\Pi$. 

We give sketch-based streaming testers for properties of being connected, $k$-edge connected, $k$-vertex connected, cycle-freeness and bipartite (for planar graphs). The performance of our testers are summarized in Table \ref{tab:1}. 
We stress that most of our testers have (asymptotically) the same post-processing time as the space they used except for testing $k$-edge connectivity when $k\ge \Omega(n^{\eps/(1+\eps)})$ and $k$-vertex connectivity when $k\ge \Omega(n^{\varepsilon/(4+\varepsilon)})$. 
\begin{table}[h]
	\centering
	\begin{tabular}{|c|c|c|c|c|c}
		\hline
		&  {Space } & {Space lower bound} \\
		& $\tilde{O}$ &   ${\Omega}$   \\
		\hline
		Connectivity& $n^{1-\eps}$& $n^{1-8\eps}$\\
		$k$-edge connectivity&$k^{1+\varepsilon}\cdot n^{1-\varepsilon}$ &\\
		$k$-vertex connectivity& $\frac{k^{1+\varepsilon/4}}{\varepsilon}\cdot n^{1-\varepsilon/4}$&\\
		Cycle-freeness&$n^{1-\varepsilon+\varepsilon^2}$& $n^{1-8\varepsilon}$\\
		Bipartiteness of planar graphs&$n^{1-\Omega(\varepsilon^2)}$&$n^{1-4\varepsilon}$\\
		\hline
	\end{tabular}
	\caption{Upper and lower bounds of streaming testers.}
	\label{tab:1}
\end{table}
\subsection{Our techniques}
To estimate the number of connected components with small additive error $\varepsilon n$, we note that it is sufficient to estimate the number $\scc(G)$ of connected components of small size (i.e., $O(1/\varepsilon)$), since the number of components of size larger than this is at most $O(\varepsilon n)$ (see also~\cite{CRT05:MST}). To estimate $\scc(G)$, the following vertex sampling framework is used: we sample a sufficient large set of vertices $S$ by sampling each vertex in $G$ with some probability $p$, and then use the statistics of the sampled connected components of the original graph to estimate $\scc(G)$. For any small connected component $C$ in $G$, it is likely that all the vertices in $C$ will be sampled out. Conditioned on this, we add $1/p^{|C|}$ to our final estimator, which is the reciprocal of the probability that $C$ is entirely sampled out.
Now the task is then to identify which subsets of $S$ are connected components in the original graph. A trivial way is to check all subsets of $S$, which takes too much time. A more efficient way is to only check all the connected components in $G[S]$, since a sampled component of $G$ must also form a component in $G[S]$. We carefully use a set of linear sketches to do this. More specifically, we first recover all connected components in $G[S]$ by invoking a sketch-based streaming algorithm given in~\cite{AGM12:linear}, which only needs space near-linear in $|S|$. Then we use (different) linear sketches to check if any of these components is indeed a connected component of the original graph. We remark that the first set of linear sketches of a vertex $v$ sketch its neighborhood information in $G[S]$, while the second set sketch its neighborhood information in $G$.
Our $o(n)$ space streaming algorithm for $(1+\varepsilon)$-approximating the weight of MST follows via a connection between the number of connected components and the weight of MST established in~\cite{CRT05:MST}.

To give testers for some graph property $\Pi$ in dynamic streaming model, we start from the observation that if a graph $G$ is far from having $\Pi$, then typically,   
%
%
%
there exist many small disjoint subgraphs, each of which is a witness that the graph $G$ does not satisfy $\Pi$. (For example, if $\Pi$ is connectivity, then there exists at least $\Omega(\varepsilon m)$ connected components of size at most $O(1/\varepsilon)$ in a graph that is $\varepsilon$-far from being connected.) This implies that by sampling a sufficient large set of vertices, with high probability, one of such subgraphs will be entirely sampled. Checking which vertices form a witness of the original graph can then be done by using the aforementioned framework. Different sketches will be used for testing different properties. 

%
%

To prove lower bounds for our studied problems, we give reductions from \emph{Boolen Hidden Hypermatching (BHH)} problem that was studied in~\cite{VY11:cycle}. Our reductions share similarity with the reduction in~\cite{VY11:cycle} to the cycle-counting problem and the reductions in~\cite{KKS15:maxcut,KK15:sketching} to the approximate max-cut problem.

\subsection{Related work}
Ahn et al.~\cite{AGM12:linear} initiated the study of graph sketches, and gave dynamic semi-streaming algorithms for computing a spanning forest (which can be used to count the exact number of connected components), and $(1+\varepsilon)$-approximate the weight of MST. They also proposed algorithms to \emph{exactly} testing of a set of properties, including testing connectivity, $k$-edge connectivity, and bipartiteness. Recently, Guha et al. gave dynamic streaming algorithms for exactly testing of $k$-vertex connectivity~\cite{GMT15:vertex}. All these algorithms use $\tilde{O}(n)$ space ($\tilde{O}(kn)$ for $k$-connectivity).
On the other hand, the randomized space lower bounds for these exact testing problems were known to be $\Omega(n)$ in the insertion-only model~\cite{FKMSZ05:graph,FKM+08:graph}. Recently, Sun and Woodruff improved these lower bounds to $\Omega(n\log n)$~\cite{SW15:bound}. Verbin and Yu~\cite{VY11:cycle} proved a lower bound for cycle-counting, which implied a lower bound of $\Omega(n^{1-O(\varepsilon)})$ for estimating the number of components.

In the \emph{random order} insertion-only model Kapralov et al. \cite{KKS14:matching} gave a one pass streaming algorithm that estimates the maximum matching size with polylogarithmic approximation ratio in polylogarithmic space. Although sublinear in $n$, the model considered is very different from ours.

Sublinear time algorithms for estimating the number of connected component and the weight of MST were first given by Chazelle et al.~\cite{CRT05:MST}. Later these two problems have been further considered in geometric settings~\cite{CEMNRS05:MST,CS09:metricMST,FIS05:sampling}. In particular, Frahling et al. studied the problem of $(1+\varepsilon)$-approximating the weight of MST in dynamic geometric data stream~\cite{FIS05:sampling}.
%

There has been a rich line of work on graph property testing in the query model~(see surveys \cite{Ron10:algorithmic,Gol11:introduction}) and the goal there is to design fast algorithms that make as few queries as possible. The query models that are mostly related to ours are bounded degree model and general graph model. In particular, our definition of $\varepsilon$-far is adapted from the general graph model. Goldreich and Ron~\cite{GR02:property} initiated the study of property testers in bounded degree graph model, and gave testers for connectivity, $k$-edge connectivity, $2,3$-vertex connectivity, cycle-freeness, Eulerianity. Testing $k$-vertex connectivity in bounded degree graphs for arbitrary constant $k$ was given in~\cite{YI08:vertex}. These testers have later been generalized to general graph model~\cite{PR02:diameter,OR11:eulerianity}. Testing bipartiteness in planar graphs was studied in~\cite{CMKS11:planar}.

After having submitted the paper, we became aware that Hossein Jowhari~\cite{Jow:estimating} has independently studied the problem of estimating the number of connected components and provided similar results as ours, while he did not consider the streaming property testers considered here. 

\section{Preliminaries}\label{sec:preliminaries}
\subsection{Notations} \label{sec:notation}
We use $[n]=\{1,\cdots, n\}$ to denote the vertex set of the graph $G$ defined by the stream, and let $m$ denote the number of edges of $G$. For an undirected graph $G=([n],E)$ and a vertex $i\in [n]$, we let $\Gamma(i)$ denote all the neighbors of $i$. For a set $C\subseteq [n]$, let $\Gamma(C)$ denote 
the set of vertices in $V\setminus C$ that have at least one neighbor in $C$, that is, $\Gamma(C)=\cup_{i\in C}\Gamma(i) \setminus C$. Let $E(C,V\setminus C)$ denote the set of edges crossing $C$ and $V\setminus C$. We will use $G[C]$ to denote the subgraph induced by $C$. 

For each vertex $i$, we define two vectors $\edgesketch^i\in\{-1,0,1\}^{\binom{n}{2}}$ and $\vertexsketch^i\in \{0,1\}^n$ to encode the neighborhood information of $i$ as follows:
\begin{equation*}
\begin{aligned}[c]
\edgesketch^i_{j,k}= \left\{\begin{array}{ll}
1 & \textrm{if $i=j<k$ and $(j,k)\in E$}\\
-1 & \textrm{if $j<k=i$ and $(j,k)\in E$}\\
0 & \textrm{otherwise}
\end{array}
\right.
\end{aligned}
\qquad
\begin{aligned}[c]
\vertexsketch^i_{j}= \left\{\begin{array}{ll}
1 & \textrm{if $j\in \Gamma(i)$ or $j=i$}\\
0 & \textrm{otherwise}
\end{array}
\right.
\end{aligned}
\end{equation*}

By simple induction arguments, it is easy to prove that for any vertex set $C\subset V$, the nonzero entries in the vector $\edgesketch^C:=\sum_{i\in C}\edgesketch^i$ corresponds to the edges between $C$ and its complement $V\setminus C$. The nonzero entries in $\sum_{i\in C}\vertexsketch^i$ corresponds exactly to vertices in $C\cup\Gamma(C)$. 
\subsection{Linear sketches}\label{sec:pre}
Linear sketch (or sketch for short) is a powerful tool widely used in the streaming model and other areas. Given a large vector $\x\in\R^n$, we want to construct a small sketch $\mathcal{L}(\x)$, from which certain properties of $\x$ can be recovered. We call $\mathcal{L}$ a linear sketch if $\mathcal{L}(\x +\y)=\mathcal{L}(\x)+\mathcal{L}(\y)$ for all $\x,\y$, and this additive property make it trivial to implement linear sketches in the dynamic streaming model.
As in the previous works, we will use linear sketches as our main tool. 



\paragraph{AGM sketch}
We will use a dynamic streaming algorithm for constructing a spanning forest of a graph by Ahn, Guha and McGregor~\cite{AGM12:linear}, which is summarized as follows. 
\begin{theorem}[AGM sketch \cite{AGM12:linear}]
	\label{thm:agm12_spanning_tree}
	There exists a single-pass sketch-based dynamic streaming algorithm that uses $O(n\log^3 n)$ space, and recovers a spanning forest of the graph with probability $0.99$. The recovery time of the algorithm is $\tilde{O}(n)$, and the update time is $\poly\log n$.
\end{theorem}

\paragraph{AMS sketch} 
To check whether the input vector $\x$ is $\0$ or not, one can simply maintain a constant approximation of its \emph{second frequency moment}, that is $F_2(\x):=\sum_i{x_i^2}$, which can be done in $O(\log n)$ space by using the classical AMS sketch that was introduced by Alon, Matias and Szegedy~\cite{AMS96:space}.


\paragraph{Exact $k$-sparse recovery} 
We call a vector $k$-sparse if $|\supp(\x)|\le k$. Given a non-zero vector $\x\in\R^n$, the goal here is to recover $\x$ if $\x$ is $k$-sparse, otherwise outputs {\bf Fail}. We have the following result from \cite{Pri11:efficient}.
\begin{lemma}[\cite{Pri11:efficient}]\label{lem:ksparse}
	There exists an $O(k\log n \log_k\delta^{-1})$ space sketch-based algorithm that takes as input a non-zero vector $\x\in\R^n$, and with probability $1-\delta$, recovers $\x$ if $\x$ is $k$-sparse, otherwise outputs {\bf Fail}. The update time is $O(\poly\log n)$ and the recovery time is $O(k\cdot\poly\log n)$.
\end{lemma}
\section{Estimating the number of connected components and MST weight}\label{sec:cc}

In this section, we present and analyze our algorithms for estimating the number of the connected components in a graph and $(1+\varepsilon)$-approximating the weight of the MST. 

\subsection{Estimating the number of connected components}
Our first observation is that, to estimate the number of connected components within additive error $\varepsilon n$, we can simply ignore all the large components (see also~\cite{CRT05:MST}). In particular, the number of components of size larger than $\Omega(1/\varepsilon)$ is at most $O(\varepsilon n)$. Thus it will be sufficient to estimate the number of components of small size, for which we have the following theorem.

\begin{theorem}\label{thm:numccsmall}
	For any constant $t\ge 1$, there exists a one-pass dynamic streaming algorithm that uses $O(e^t n^{1-\varepsilon}\cdot\poly\log n)$ space and post-processing time to estimates the number of connected components of size at most $1/\varepsilon$ within an additive error $\varepsilon^t n$. The update time is $O(\poly\log n)$.
\end{theorem}

By invoking Theorem \ref{thm:numccsmall} with parameter $\varepsilon'=(1-\varepsilon^q)\varepsilon$ and $t=(q+1)$, we get an estimator for the number of connected components of size smaller than $1/\varepsilon'$ within additive error at most $\varepsilon^{q+1}n$. Since the number of components of size at least $1/\varepsilon'$ is at most $\varepsilon' n = \varepsilon n-\varepsilon^{1+q}n$, the estimator also approximates the total number of connected components within additive error at most $\varepsilon n$. The space of the algorithm is $\tilde{O}(e^{q+1}n^{1-\varepsilon+\varepsilon^{q+1}})$, and we have the following result.

\begin{theorem}~\label{thm:numcc}
	Let $q\ge 1$ be a constant. There exists a one-pass dynamic streaming algorithm that with constant success probability, estimates the number of connected components of a graph within an additive error $\varepsilon n$ in $O(e^{q+1} n^{1-\varepsilon+\varepsilon^{q+1}}\cdot\poly\log n)$ space and post-processing time.
\end{theorem}

Now we give the proof of Theorem \ref{thm:numccsmall}. Recall that the vectors $\edgesketch^C$ encode the information of the number of edges between $C$ and $V\setminus C$.
\begin{proof}[Proof of Theorem \ref{thm:numccsmall}.]
	Let $\scc(G)$ denote the number of connected components of size at most $1/\varepsilon$ in $G$. Our algorithm for estimating $\scc(G)$ is as follows. We first sample each vertex with probability $p:=(\varepsilon^{2t} n/16)^{-\varepsilon}$. Let $S$ be the set of sampled vertices. We then use the AGM sketch from Theorem \ref{thm:agm12_spanning_tree} to maintain a spanning forest $F$ of the subgraph induced by $S$. 
	Then for each component $C$ in $F$, we test whether $C$ is actually a connected component in $G$ by testing whether the vector $\edgesketch^C:=\sum_{v\in C}\edgesketch^v$ is $\0$, which can be done by the AMS sketch. If $\edgesketch^C=\0$, we set $X_C=1$, otherwise set $X_C=0$. Our estimator is then defined as
	$\sum_{C} \frac{X_C}{p^{|C|}}$, where $C$ ranges over all components of $F$ with size at most $\frac{1}{\varepsilon}$. See {\bf Algorithm} \ref{alg:number_connected_components} for the details.
	
	\begin{algorithm}[h]
		\caption{\texttt{EstimateNumSCC}}
		\label{alg:number_connected_components}
		\begin{algorithmic}[1]
			\State	
			Sample each vertex with probability $p:=(\varepsilon^{2t} n/16)^{-\varepsilon}$. If more than $16 np$ vertices are sampled, then abort and output {\bf Fail}. Let $S$ denote the set of sampled vertices.~\label{alg:numcc_sample_vertices}
			\State Maintain an AGM sketch of $G[S]$ using Theorem~\ref{thm:agm12_spanning_tree}.  
			\State For each $v\in S$, maintain an AMS sketch $AMS(\edgesketch^v)$, sketching the neighborhood of $v$ in $G$.
			\State {\bf Post-Processing:} 
			\State 
			Use the AGM sketch to recover a spanning forest $F$ of $G[S]$ using Theorem \ref{thm:agm12_spanning_tree}. 
			\State For each component $C\in F$, estimate $F_2(\edgesketch^C)$ using the AMS sketch  $AMS(\edgesketch^C)=\sum_{v\in C}AMS(\edgesketch^v)$, and set $X_C=1$ if $F_2=0$, otherwise set $X_C=0$. For each $1\leq \ell\leq \frac{1}{\varepsilon}$, let $X_\ell:=\sum_{C:|C|=\ell}X_C$.
			\label{alg:numcc_detect_ccs}
			\State Output $Y:=\sum_{\ell\leq \frac{1}{\varepsilon}} \frac{X_\ell}{p^\ell}$.
		\end{algorithmic}
	\end{algorithm}
	
	Note that the algorithm samples at most $16np=O(\varepsilon^{-2t\varepsilon}\cdot n^{1-\varepsilon})$ vertices and we maintained an AGM sketch on $G[S]$ and an AMS sketch for each sampled vertex, which imply that the space complexity of the algorithm is $O(\varepsilon^{-2t\varepsilon} n^{1-\varepsilon}\cdot\poly\log n)$. By simple calculus, for any $\varepsilon$, it holds that $\varepsilon^{-2\varepsilon}\le e^{2/e}<e$, so the space is at most  $\tilde{O}(e^tn^{1-\varepsilon})$. The post-processing time is near linear in the space, and the update time is $O(\poly\log n)$.
	
	Now we prove the correctness of the above algorithm. First we note that the expected number of sampled vertices in Step~(\ref{alg:numcc_sample_vertices}) is $np$, and thus by Markov inequality, the probability that more than $16 np$ vertices are sampled is at most $\frac{1}{16}$. Also note that with probability at least $1-\frac{1}{16}$, the AGM sketch returns a true spanning forest of $G[S]$. In addition, since the number of components in $F$ is at most $n$, we will query the AMS sketch at most $n$ times. Thus if we set the error probability of the AMS sketch to be $\frac{1}{16n}$ (with an extra $\log n$ factor in space), then with probability at least $1-\frac{1}{16}$, all invocations of AMS sketches for testing if $\edgesketch^C = \0$ will give the correct answer. Conditioned on this event, $X_\ell$ defined in Step~(\ref{alg:numcc_detect_ccs}) is exactly the number of connected components $B$ of size $\ell$ in $G$ such that all vertices in $B$ are sampled out, which is true since for any component $C\in F$, $F_2(\edgesketch^C)=\0$ if and only if $C$ is a connected component in $G$. 

	Let $B_1,\cdots,B_{\scc(G)}$ be the connected components of size at most $\frac{1}{\varepsilon}$ of $G$. 
	For any integer $\ell\le\frac{1}{\varepsilon}$, let $\BB_\ell$ denote the set of connected components of size $\ell$ in $G$, that is, $\BB_\ell=\{B_i: 1\leq i\leq \scc(G), \abs{B_i}=\ell\}$. Let $b_\ell:=|\BB_\ell|$. Note that $\scc(G)=\sum_{\ell\leq \frac{1}{\varepsilon}}b_\ell$. For any set $B$, let $Z_B$ denote the indicator random variable that all the vertices in $B$ have been sampled. Note that $\Pr[Z_B=1] = p^{\abs{B}}$. Now by the above argument, $X_\ell=\sum_{B\in\BB_\ell}Z_B$, and $\E[X_\ell] = b_\ell\cdot p^\ell$. Furthermore, we have
	$Y=\sum_{\ell\leq \frac{1}{\varepsilon}} \frac{X_\ell}{p^\ell}=\sum_{\ell\leq \frac{1}{\varepsilon}} \frac{\sum_{B\in\BB_\ell}Z_B}{p^\ell},$ and thus $\E[Y]=\sum_{\ell\leq \frac{1}{\varepsilon}}b_\ell =\scc(G)$.
	
	Note that all $Z_{B_i}$'s are mutually independent for all $i$, so it holds that
	\begin{eqnarray}
	\Var[Y]&=&\sum_{\ell\leq \frac{1}{\varepsilon}}\frac{\sum_{B\in\BB_\ell}\Var[Z_B]}{p^{2\ell}} \nonumber
	= \sum_{\ell\leq \frac{1}{\varepsilon}}\frac{b_\ell(p^{\ell} - p^{2\ell})}{p^{2\ell}} \nonumber
	\leq \sum_{\ell\leq \frac{1}{\varepsilon}}\frac{b_\ell}{p^{\ell}} \\
	&\leq& \frac{\sum_{\ell\leq \frac{1}{\varepsilon}}b_\ell}{p^{{1}/{\varepsilon}}} = \frac{\scc(G)}{p^{1/\varepsilon}} 
	\le \frac{n}{p^{{1}/{\varepsilon}}}= \varepsilon^{2t}n^2/16,\label{eqn:var}
	\end{eqnarray} 
	where we use the fact that $\scc(G)\leq n$, and  $p=(\varepsilon^{2t} n/16)^{-\varepsilon}$.
	Then by Chebyshev's inequality, 
	\begin{eqnarray*}
		\Pr[|Y-\scc(G)|\geq \varepsilon^t n]=\Pr[|Y-\E[Y]|\geq \varepsilon^t n] \le \frac{\Var[Y]}{\varepsilon^{2t} n^2} \leq 1/16.
	\end{eqnarray*}
	
	%
	%

	By the union bound, the algorithm will succeed with probability at least $\frac23$. 
	%
	%
	%
\end{proof}
\subsection{Approximating the weight of minimum spanning tree}
We use the previous algorithm on estimating the number of connected components to approximate the weight of a minimum spanning tree of a weighted graph. Let $W\geq 2$ be an integer, $G$ be a connected graph with integer edge weights from $[W]:=\{1,\cdots,W\}$, and  $c(\MST)$ be the weight of an MST of $G$. For any $1\leq \ell\leq W$, let $G^{(\ell)}$ denote the subgraph of $G$ consisting of all edges of weight at most $\ell$. Let $\cc^{(\ell)}$ denote the number of connected components of $G^{(\ell)}$. 
Chazelle et al.~\cite{CRT05:MST} give the following lemma relating the weight of MST to the number of connected components of $G^{(\ell)}$. 
\begin{lemma}[\cite{CRT05:MST}]\label{lem:MSTCC}
	It holds that $c(\MST)=n - W + \sum_{\ell=1}^{W-1}\cc^{(\ell)}$.
\end{lemma}

For a connected graph with integer edge weights, the weight of any MST is at least $n-1$, so it is sufficient to estimate $cc^{(\ell)}$ within an additive error of $\varepsilon n/(W-1)$ for each $\ell$. To do this, we can simply run $W-1$ parallel instances of Theorem \ref{thm:numcc}, each of which sketches a subgraph $G^{(\ell)}$. Then the space of the algorithm will be $\tilde{O}(Wn^{1-\frac{\varepsilon}{W-1}})$. 

\begin{theorem}\label{thm:upper_mst}
	Let $t\ge 1$ be any constant. There exists a single-pass dynamic streaming algorithm that uses space and post-processing time $O(e^tW n^{1-\frac{\varepsilon}{W-1}+\frac{\varepsilon^t}{(W-1)^t}}\poly\log n)$ to compute a $(1+\varepsilon)$-approximation of the weight of the MST.
\end{theorem}

We remark that Ahn et al.~\cite{AGM12:linear} have given a dynamic streaming algorithm for this problem for any graph with maximum edge weight upper bounded by $O(\poly(n))$, and their algorithm uses space $O(n\cdot\poly\log n)$. Our algorithm uses $o(n)$ space for any connected graph with maximum edge weight bounded by $o(\log n)$ (for constant $\varepsilon$), which improves the algorithm of~\cite{AGM12:linear} in this setting. We also note that $\Omega(n)$ space is necessary for estimating the weight of MST for graphs with maximum edge weight at least $c\log n$ for constant $\varepsilon$ and some large universal constant $c$ (see Theorem~\ref{thm:lower_mst}). 
Finally, we remark that the algorithm can also be extended to the setting where non-integral weights are allowed (see \cite{CRT05:MST} for more details).

\section{Dynamic streaming testers}
In this section, we give our streaming testers for a number of graph properties, including k-edge
connectivity, cycle-freeness, and planar graph bipartiteness. We present the testers for k-vertex
connectivity and Eulerianity in Appendix~\ref{app:otherProperties}.  

\subsection{Testing $k$-edge connectivity}
A graph is $k$-edge connected if the minimum cut of the graph has size at least $k$. We start from the simplest case, i.e., $k=1$, which is equivalent to the problem of testing connectivity.
\subsubsection{Connectivity}

It is clear that if $G$ is $\varepsilon$-far from being connected, one must add at least $\varepsilon m$ edges to make it connected, which implies that there are at least $\varepsilon m + 1$ connected components in $G$ \cite{GR02:property,PR02:diameter}. Therefore, we can also solve this by estimating the number connected components by setting the error parameter appropriately, however, by a more careful analysis, we can improve this by a factor of $O(n^{O(\varepsilon)})$.
\begin{theorem}\label{thm:connectivity}
	There exists a dynamic streaming tester for  $1$-edge connectivity that runs in $\tilde{O}(n^{1-\varepsilon})$ post-processing time and space.
\end{theorem}
\begin{proof}
	First observe that one can simply reject the input graph if $m<n-1$, since in this case, the graph is disconnected. Thus, in the following we assume $m\ge n-1$ and our tester is described in {\bf Algorithm} \ref{alg:test_connectivity}. 
	
	\begin{algorithm}
		\caption{TestConnectivity}
		\label{alg:test_connectivity}
		\begin{algorithmic}[1]
			\State
			Sample each vertex with probability $p:=( \varepsilon n/10)^{-\varepsilon}$. If more than $16np$ vertices are sampled, abort and output {\bf Fail}. Let $S$ denote the set of sampled vertices.
			\State For each $v\in S$, maintain an AMS sketch $AMS(\edgesketch^v)$, sketching the neighborhood of $v$ in $G$.
			\State Maintain an AGM sketch of $G[S]$ using Theorem~\ref{thm:agm12_spanning_tree}. 
			\State {\bf Post-Processing:} 
			\State Use the above sketch to construct a spanning forest $F$ of  $G[S]$ as guaranteed by Theorem~\ref{thm:agm12_spanning_tree}.
			
			\State For each connected component $C\in F$, estimate $F_2(\edgesketch^C)$ using the AMS sketch  $AMS(\edgesketch^C)=\sum_{v\in C}AMS(\edgesketch^v)$. If the answer $\tilde{F_2}=0$, {\bf Reject}.
			\State {\bf Accept}.
		\end{algorithmic}
	\end{algorithm}
	
	It is easy to see that {\bf Algorithm} \ref{alg:test_connectivity} only use $\tilde{O}(|S|)$ space, which is bounded by $\tilde{O}(np)=\tilde{O}(\varepsilon ^{-\varepsilon} n^{1-\varepsilon})=\tilde{O}(n^{1-\varepsilon})$. 
	The post-processing time is nearly linear in the size of $S$, since the AGM algorithm needs $\tilde{O}(|S|)$ post-processing time, and we invoke at most $|S|$ AMS queries, each of which takes $\tilde{O}(1)$ time. The update time is $\poly\log n$.

	For the correctness of the algorithm, we condition on the event that the number of sampled vertices is at most $16np$, which occurs with probability at least $1-\frac{1}{16}$, and on the event that the spanning forest $F$ is constructed correctly, which occurs with probability $0.99$. By setting the error probability of the AMS sketch to be $1/n^2$ (with an extra $\log n$ factor in space), with probability $0.99$, all the answers from AMS sketches are all correct, and we also condition on this. 
	
	If $G$ is connected, then it will always be accepted, since for each $C\in F$, $\edgesketch^C\neq \0$, and conditioned on the correctness of the AMS sketch, $\tilde{F_2}$ will never be $0$. 
	On the other hand, if the graph is $\varepsilon$-far from being connected, the number of connected components in $G$, denoted as $\cc(G)$, is at least $1+\varepsilon m \ge \varepsilon n$. 
	Let $B_1,\cdots,B_{\cc(G)}$ denote all connected components in $G$. Let $p_i=p^{|B_i|}$ for $1\le i\le \cc(G)$. Using the inequality $1-x\le e^{-x}$ for all $x$, the probability that none of the components is entirely sampled out is $(1-p_1)\cdot (1-p_2)\cdots \cdot (1-p_{\cc(G)})\le  e^{-\sum_i p_i}$. Then by the AM-GM inequality, this probability is at most
	\begin{eqnarray}
	e^{-\cc(G)\cdot (\prod_i p_i)^{1/\cc(G)}}=e^{-\cc(G)\cdot p^{n/\cc(G)}}\le e^{-\cc(G)\cdot p^{1/\varepsilon}}\le e^{-\varepsilon n \cdot p^{1/\varepsilon}}\le 1/16, \nonumber
	\end{eqnarray}
	where we use the fact that $p=(\varepsilon n/10)^{-\varepsilon}$ and $\cc(G)\geq \varepsilon n$. So the probability that at least one of the components is sampled out is at least $15/16$. Conditioned on this, $F_2(\edgesketch^C)=0$ for some component in $G[S]$ and the algorithm will output {\bf Reject}. By union bound, our algorithm will succeed with probability $1-\frac{1}{16}-0.01-0.01-\frac{1}{16}>3/4$.
	%
	%
	%
\end{proof}
\subsubsection{$k$-edge connectivity: $k\geq 2$}~\label{sec:k-edge-connectivity}
By using a slightly more involved argument and replacing AMS sketches with $(k-1)$-sparse recovery sketches, we can generalize the above idea to testing $k$-edge connectivity for $k\geq 2$. We have the following theorem on testing $k$-edge connectivity. The proof is deferred in Appendix~\ref{app:edge-conn}. 
\begin{theorem}~\label{thm:kedgeconn}
	Let $k\le O(n^{\varepsilon/(1+\varepsilon)})$. There exists a single-pass dynamic streaming tester for $k$-edge connectivity with post-processing time and space ${O}(k^{1+\varepsilon}\cdot n^{1-\varepsilon}\cdot \poly\log n)$.
\end{theorem}

\subsection{Testing cycle-freeness}\label{sec:cyclefree}
Now we consider the problem of testing cycle-freeness, which is equivalent to testing if the graph is a forest. Let $\cc(G)$ denote the number of connected components of the input graph $G$. Let $B_1,\cdots, B_{\cc(G)}$ be the connected components in $G$. Note that if $G$ is cycle-free, then for each $i\leq \cc(G)$, $|E(B_i)|=|B_i|-1$, and thus the total number of edges in $G$ is 
\begin{eqnarray*}
	m= \sum_{i=1}^{\cc(G)}|E(B_i)|=\sum_{i=1}^{\cc(G)}(|B_i|-1)=n-\cc(G),
\end{eqnarray*} 
that is, $\cc(G)=n-m$. If $G$ is $\varepsilon$-far from being cycle-free, i.e.,  one has to delete more than $\varepsilon m$ edges to make it cycle-free, then $\cc(G)> n-m+\varepsilon m$. Therefore, to test cycle-freeness of a graph, it will be sufficient to approximate the number of connected components with additive error $\varepsilon m/2$. One may try to directly invoke {\bf Algorithm} \ref{alg:number_connected_components} with parameter $\varepsilon'=\frac{\varepsilon m}{2n}$. However, $m$ could be much smaller than $n$ and we do not know $m$ in advance. We overcome this obstacle by a case analysis. 
%
%
%

\begin{theorem}
	There exists a single-pass dynamic streaming algorithm that tests cycle-freeness of a graph with space and post-processing time $O(n^{1-\varepsilon+\varepsilon^2}\cdot\poly\log n)$.
\end{theorem}
\begin{proof}
	Note that if $m>n-1$, then the graph must contain at least one cycle, and thus we can safely reject the graph. In the following, we assume that $m\le n-1$.
	Our algorithm for testing cycle-freeness depends on the construction of AGM sketch, in which each vertex $u$ maintains a linear sketch of $\edgesketch^{u}$ (denoted as $\sketch(\edgesketch^u)$). Each such sketch has size $\poly\log n$ and the property that $\sketch(\0)=\0$ (it consists of $O(\log n)$ $l_0$-samplers, see \cite{AGM12:linear} for details). Our main idea is to maintain a sparse recovery sketch for the AGM sketch (i.e. a composition of sparse recovery sketch and AGM sketch). Now we describe our algorithm as follows.
	\begin{algorithm}
		\caption{TestCycleFreeness}
		\label{alg:test_cyclefree}
		\begin{algorithmic}[1]
			\State
			Maintain a count $\lambda$ of the number of edges. 
			\State\label{alg:cycle_small} Let  $\eta={\varepsilon}/{(1+\varepsilon+\varepsilon^2)}$. Let $k=n^{1-\eta}\poly\log n$. Maintain an exact $k$-sparse recovery sketch $\mathcal{S}$ of the vector $\sparsesketch:=(\sketch(\edgesketch^u))_{u\in V}$ using Lemma \ref{lem:ksparse}. 
			\State \label{step:largem} Run {\bf Algorithm} \ref{alg:number_connected_components} with parameter $p=(2^{2t}n^{1-\eta}/16)^{-\eta}$, while in step $(\ref{alg:numcc_detect_ccs})$ of {\bf Algorithm} \ref{alg:number_connected_components}, ignore all the isolated vertices that are sampled out (i.e., set $X_C=0$ whenever $|C|=1$).
			\State {\bf Post-Processing:} 
			\State Recover $\sparsesketch$ from $\mathcal{S}$. \If{ The recovery does not fail} 
			\State  Use $\sparsesketch$ to construct a spanning forest on vertex set $Y:=\{u:\sketch(\edgesketch^u)\neq \0\}$ using Theorem~\ref{thm:agm12_spanning_tree}. Let $\tilde{c}_1$ denote the number of connected components of this forest. If $\tilde{c}_1= |Y|-\lambda$, \textbf{Accept}; otherwise, \textbf{Reject}.
			\Else\label{alg:cycle_large} 
			\State Let $\tilde{c}_2$ be the resulting estimator of {\bf Algorithm} \ref{alg:number_connected_components} in Step~\ref{step:largem}. If $\tilde{c}_2\le n-(1-\varepsilon -\frac{\varepsilon^3}{4})\lambda$, {\bf Accept}; otherwise, {\bf Reject}.
			\EndIf
		\end{algorithmic}
	\end{algorithm}
	
	Note that the space used by the algorithm is $\max\{\tilde{O}(np), k\cdot\poly\log n\} = \tilde{O}(n^{1-\varepsilon+\varepsilon^2})$, and the post-processing time is near linear in space.
	
	Now we prove the correctness of the algorithm. We define $G'\subseteq G$ to be a subgraph which consists of all the vertices of positive degree. Let $n'=|G'|$. Note that $m\ge n'/2$. 
	
	If $n'\le n^{1-\eta}$, then the vector $\sparsesketch$ is $\tilde{O}(n^{1-\eta})$-sparse, since for all isolated vertices $u$, we have $\sketch(\edgesketch^u)=\0$, and thus we can recover the entire $\sparsesketch$ exactly. Then by Step (\ref{alg:cycle_small}) and Theorem \ref{thm:agm12_spanning_tree}, we can get the exact number of components of $G'$. Since the number of vertices of $G'$ is $|Y|$, and $\lambda = m$ is the total number of edges, then the graph is cycle-free if and only if $\tilde{c}_1= |Y|-\lambda$. 
	
	If $n'>n^{1-\eta}$, then by Theorem \ref{thm:numccsmall}, $\tilde{c}_2$ is an estimator for the number of components in $G'$ of size smaller than $1/\eta$ with additive error $\eta^t\sqrt{n'n^{1-\eta}}$. This follows by the upper bound $\eta^{2t}n^{1-\eta} n'/16$ of the variance of the estimator (which can be shown similarly to inequality (\ref{eqn:var}) in Section \ref{sec:cc}) and the Chebyshev's inequality. Now note that the additive error is at most $\eta^t n'\le \varepsilon^3 m/8$ for some constant $t$ since $n'>n^{1-\eta}$ and $m\ge n'/2$. Let $L$ be the number of components in $G'$ of size larger than $1/\eta$, then $-\varepsilon^3 m/8 \le \cc(G')-\tilde{c}_2\le L+\varepsilon^3 m/8$ holds with high probability. Conditioned on this, Step~(\ref{alg:cycle_large}) outputs the correct answer if $L+\varepsilon^3 m/8+\varepsilon^3 m/8 = L+\varepsilon^3 m/4 < \varepsilon m$. Now if $L< \varepsilon m/4$, we are done. If $L \ge \varepsilon m/4$, then by our choice of $\eta$ and the fact that $m\ge L\cdot(1/\eta -1)$, $\varepsilon m\ge L+L\varepsilon^2\ge L+\varepsilon^3 m/4$. This completes the proof of the theorem.
	%
	%
	%
\end{proof}

\subsection{Testing bipartiteness of the planar graphs}\label{app:testingBipartite} 
Now we consider the problem of testing if a planar graph is bipartite or $\varepsilon$-far from  bipartite. Here a planar graph is $\varepsilon$-far from bipartite if one has to delete at least $\varepsilon m$ edges to get a bipartite graph.
Czumaj et al.~\cite{CMKS11:planar} showed the following result\footnote{In~\cite{CMKS11:planar}, $\varepsilon$-far is expressed as $\varepsilon n$ edges, rather than $\varepsilon m$ edges as in our definition, that has to be deleted to obtain a bipartite graph. However, Lemma~\ref{lemma:planar_bipartite} directly follows from their proof.}. 
\begin{lemma}[\cite{CMKS11:planar}]~\label{lemma:planar_bipartite}
	For any (simple) planar graph $G$ that is $\varepsilon$-far from bipartite, then $G$ has at least $\varepsilon m/q(\varepsilon)$ edge-disjoint odd-length cycles of length at most $q(\varepsilon)/2$ each, where $q(\varepsilon)=O(1/\varepsilon^2)$.
\end{lemma}
By the above lemma, we only need to sample each \emph{edge} independently with some probability (rather than vertices as we did before) of the graph so that with high probability the resulting sampled graph contains at least one short odd-length cycle. The edge-sampling process can be done by using hash functions (see e.g.~\cite{AGM12:sketch}). Similar to our previous analysis, it will be sufficient to set the sample probability to $p=O_\varepsilon( n^{-q(\varepsilon)})$, which implies that the space used is $O(n^{1-\Omega(\varepsilon^2)}\poly\log n)$. We omit the details here. 

\section{Lower bounds}
In this section we present lower bounds, which hold in the insertion-only model. Our proofs are based on the reductions to the \emph{Boolean Hidden Hypermatching (BHH)} problem (See \cite{VY11:cycle}), which are in the same spirit as the lower bound proof for the \emph{Cycle Counting} problem in \cite{VY11:cycle}. We first give the definition of the boolean hidden hypermatching problem.
\begin{definition}[$\BHH_n^t$]
	In the this problem, Alice gets a boolean vector $x\in\{0,1\}^n$, where $n=2kt$ for some integer $k$. Bob gets a partition (or hypermatching) of the set $[n]$, $\{m_1,\cdots,m_{n/t}\}$, where the size of each $m_i$ is $t$, and a vector $w\in\{0,1\}^{n/t}$. For convenience, we will also use the corresponding $n$-dimensional boolean indicator vector $M_i$ to represent $m_i$, and let $M$ be a $n/t\times n$ matrix, the $i$ row of which is $M_i$. The promise of the input is either $Mx+w=\textbf{1}$ or $Mx+w=\textbf{0}$, where all the operations are modulo $2$. The goal of the problem is to output $1$ when $Mx+w=\textbf{1}$, and output $0$ otherwise.
\end{definition} 

We have the following lower bound from \cite{VY11:cycle}.
\begin{theorem}[\cite{VY11:cycle}]
	The randomized one-way communication complexity of $\BHH_n^t$ when $n=2kt$ for some integer $k\ge 1$ is $\Omega(n^{1-1/t})$.
\end{theorem}
Our lower bounds will be built upon the following basic construction.

\paragraph{Construction of $G(x,M)$.} Given vector $x$ and matrix $M$ respectively, Alice and Bob construct a bipartite graph $G(x,M)=(U,V,E)$, where $U=\{u_1,\cdots, u_{2n}\}$ and $V=\{v_1,\cdots, v_{2n}\}$, as follows.
Given $x\in\{0,1\}^n$, Alice adds a perfect matching between $U$ and $V$. 
For each $i\in[n]$, if $x_i=0$, she adds two parallel edges $(u_{2i-1},v_{2i-1})$ and $(u_{2i},v_{2i})$; otherwise if $x_i=1$, she adds two crossing edges $(u_{2i-1},v_{2i})$ and $(u_{2i},v_{2i-1})$ (see Figure~\ref{fig:alice}). 
\begin{figure}[htb]
	\centering	        
	\includegraphics[width=0.6\linewidth]{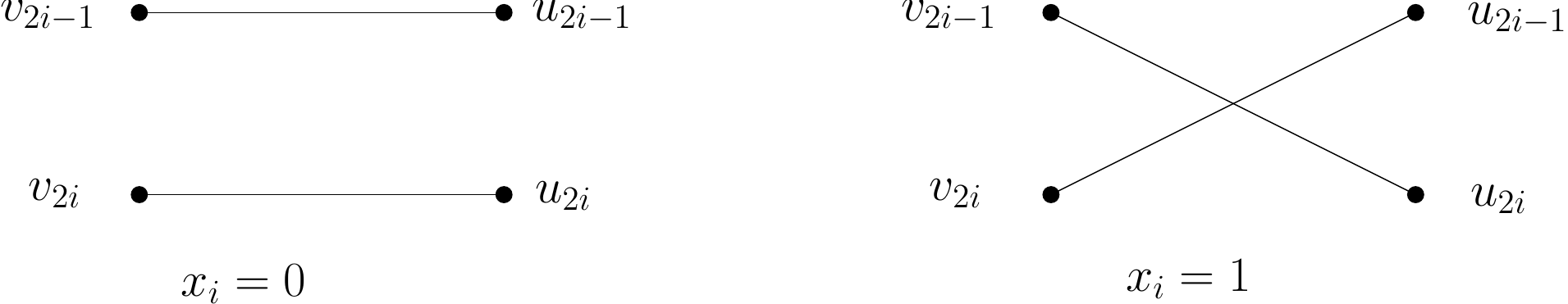}		
	\caption{Parallel (left) and crossing (right) matching according to the value of $x_i$}
	\label{fig:alice}
\end{figure}

Given $M$, Bob will do the following. For each $i\in [n/t]$ and the hyperedge $m_i\subset [n]$ (that corresponds to the $i$th row $M_i$), we use $m_{i,j}\in[n]$ to denote the $j$th element in $m_i$ and we let $S_i:=\{w| w= \textrm{$v_{2m_{i,j}-1}$ or $v_{2m_{i,j}}$ or $u_{2m_{i,j}-1}$ or $u_{2m_{i,j}}$}, j \in [t]\}$. 
For each $i\in [n/t]$ and $j\in [t-1]$, Bob adds two edges $(u_{2m_{i,j}-1},v_{2m_{i,j+1}-1})$ and $(u_{2m_{i,j}},v_{2m_{i,j+1}})$ 
(See Figure~\ref{fig:bob}). 
\begin{figure}[htb]
	\centering	        
	\includegraphics[width=0.6\linewidth]{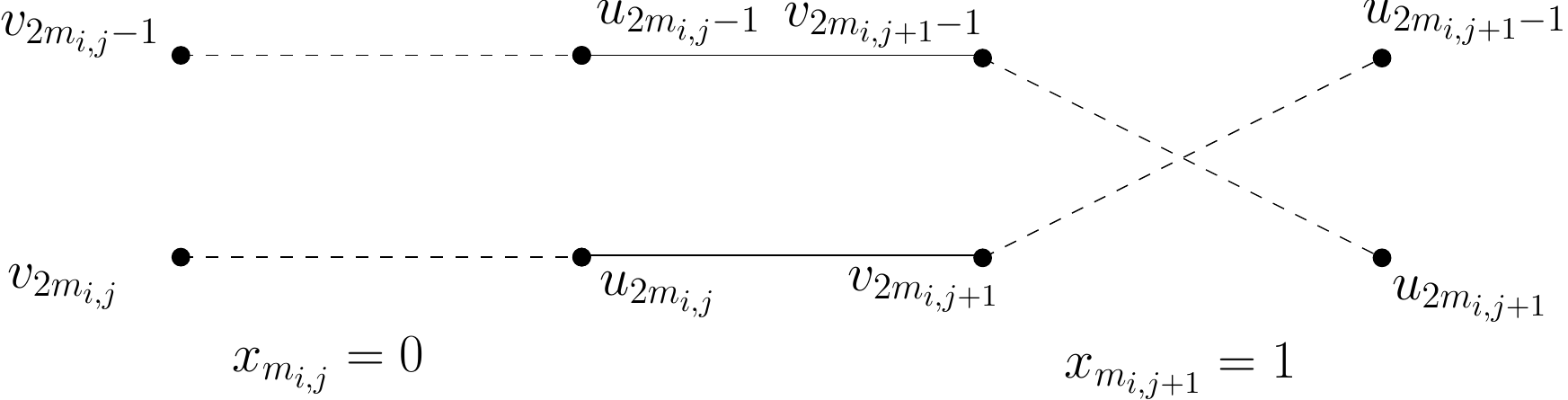}		
	\caption{Bob connects $(u_{2m_{i,j}-1},v_{2m_{i,j+1}-1})$ and $(u_{2m_{i,j}},v_{2m_{i,j+1}})$ for each $j\in [t-1]$}
	\label{fig:bob}
\end{figure} 

Observe that the edges added by Alice and Bob form two paths $p_{2i-1},p_{2i}$ over vertex set $S_i$, where $p_{2i-1}$ starts from $v_{2m_{i,1}-1}$ and $p_{2i}$ starts from $v_{2m_{i,1}-1}$ for each $i$. The entire graph $G(x,M)$ consists of $2n/t$ disjoint paths $\{p_1\cdots,p_{2n/t}\}$. It also has the following property.

\begin{fact}
Based on the value of $(Mx)_i$, we have: 1) if $(Mx)_i=0$, then $p_{2i-1}$ is a path from $v_{2m_{i,1}-1}$ to $u_{2m_{i,t}-1}$ and $p_{2i}$ is a path from $v_{2m_{i,1}}$ to $u_{2m_{i,t}}$; 2) if $(Mx)_i=1$, then $p_{2i-1}$ is a path from $v_{2m_{i,1}-1}$ to $u_{2m_{i,t}}$ and $p_{2i}$ is a path from $v_{2m_{i,1}}$ to $u_{2m_{i,t}-1}$.
\end{fact}

\subsection{Minimum spanning tree}
\begin{theorem}\label{thm:lower_mst}
	In the insertion-only model, if all edges of the graph have weights in $[W]$, any algorithm that $(1\pm \varepsilon)$-approximates the weight of the MST must use $\Omega(n^{1-\frac{4\varepsilon}{W-1}})$ bits of space.
\end{theorem}
\begin{proof}
	Given $x$ and $M$, Alice and Bob first construct the graph $G(x,M)$ as describe above. Next Bob adds $(u_{2m_{i,t}-1},v_{2m_{i,1}-1})$ and $(u_{2m_{i,t}},v_{2m_{i,1}})$ if $w_i=0$; adds $(u_{2m_{i,t}-1},v_{2m_{i,1}})$ and $(u_{2m_{i,t}},v_{2m_{i,1}-1})$ if $w_i=1$. The weight of all the edges added so far is $1$. Next Bob places edges $(v_{2m_{i,t}},v_{2m_{i+1,1}})$ with weight $1$ for $i=1,\cdots,n/t-1$ and edges $(v_{2m_{i,t}},u_{2m_{i,t}})$ with weight $W$ for each $i\in[n/t]$, so that the graph become connected. By similar argument as above, if $Mx+w=\textbf{0}$, all the edges $(v_{2m_{i,t}},u_{2m_{i,t}})$ must be picked in any minimum spanning tree, since each of these edges forms a cut, and thus the weight of any MST is $nW/t+4n-n/t-1=4n\varepsilon+4n-1$, where we set $t=(W-1)/4\varepsilon$. On the other hand, when $Mx+w=\textbf{1}$, the weight of the MST is $4n-1$, since in this case, the graph is already connected without those edges with weight $W$. So if the algorithm can compute an $(1+\varepsilon)$-approximation of the weight of the minimum spanning tree, it solves the $\BHH_n^t$ problem. This completes the proof.
\end{proof}
\subsection{Testing connectivity}
\begin{theorem}\label{thm:lowerconnec}
	In the insertion-only model, to distinguish whether a graph of $4n$ vertices is connected or $\frac{1}{8t+1}$-far from being connected, any algorithm must use $\Omega(n^{1-1/t})$ bits of space.
\end{theorem}
\begin{proof}
	Given $x$ and $M$, Alice and Bob first construct the graph $G(x,M)$.
	Next Bob adds another set of edges based on vector $w$. If $w_i=0$, he adds $(u_{2m_{i,t}-1},v_{2m_{i,1}-1})$ and $(u_{2m_{i,t}},v_{2m_{i,1}})$; if $w_i=1$, he adds $(u_{2m_{i,t}-1},v_{2m_{i,1}})$ and $(u_{2m_{i,t}},v_{2m_{i,1}-1})$. So when $(Mx)_i+w_i=0$, $p_{2i-1}$ and $p_{2i}$ become $2$ disjoint cycles. On the other hand, when $(Mx)_i+w_i=1$, $p_{2i-1}$ and $p_{2i}$ together form a larger cycle.
	Now Bob places $(v_{2m_{i,t}},v_{2m_{i+1,1}})$ in $E$ for $i=1,\cdots,n/t-1$ which connect $p_{2i}$ with $p_{2(i+1)}$ for all $i\in[n/t-1]$, i.e. all the paths in $G(x,M)$ with even indices become a connected component. The total number of edges is $8n+n/t$. When $Mx+w=\textbf{0}$, the graph has $n/t+1$ components which is $\frac{1}{8t+1}$-far from connected; when $Mx+w=\textbf{1}$ the graph is connected. So if a streaming algorithm can
	distinguish whether a graph of size $4n$ is connected or $1/8t$-far from being connected, it solves $\BHH_n^t$, since Alice can first run the algorithm on her part of the graph and send the memory to Bob, and then Bob continues to run the algorithm on his part and output the answer. Therefore, the communication lower bound of $\BHH_n^t$ implies a space lower bound of testing connectivity. 
\end{proof}

\subsection{Testing cycle-freeness}
As in the proof of Theorem~\ref{thm:lowerconnec}, given $x$ and $M$, Alice and Bob first construct $G(x,M)$. Then, for $i\in[n/t]$, Bob adds $(u_{2m_{i,t}-1},v_{2m_{i,1}-1})$ if $w_i=0$; adds $(u_{2m_{i,t}-1},v_{2m_{i,1}})$ if $w_i=1$. The total number of edges is less than $8n$. Through similar arguments, it is easy to verify that if if $Mx+w=\textbf{0}$, the graph has exactly $n/t$ cycles and $n/t$ paths, which is $1/8t$-far from cycle-free. On the contrary, if $Mx+w=\textbf{1}$, the graph has $n/t$ paths and no cycle. So if an algorithm can
distinguish whether a graph of size $4n$ is cycle-free or $1/8t$-far from cycle-free, it solves $\BHH_n^t$.
\begin{theorem}
	In the insertion-only model, any algorithm that can distinguish whether a graph of $4n$ vertices is cycle-free or $1/8t$-far from being cycle-free, must use $\Omega(n^{1-1/t})$ bits of space.
\end{theorem}

\subsection{Testing bipartiteness of planar graphs}
Alice and Bob first construct the graph $G(x,M)$. Next, for each $i\in[n/t]$, Bob adds edges $(v_{2m_{i,1}-1}, \xi_1)$ and $(v_{2m_{i,1}}, \xi_2)$, where $\xi_1,\xi_2$ are new vertices.  For $i\in[n/t]$, Bob also adds $(u_{2m_{i,t}-1},\xi_1)$ and $(u_{2m_{i,t}},\xi_2)$ if $w_i=0$; adds $(u_{2m_{i,t}-1},\xi_2)$ and $(u_{2m_{i,t}},\xi_1)$ if $w_i=1$. For this problem we assume $t$ is odd. So by similar arguments, we can easily verify that, if $Mx+w=\textbf{0}$, the graph contains $2n/t$ edge-disjoint cycles of length $2t+1$, and if $Mx+w=\textbf{1}$, the graph has no odd cycle, and thus bipartite. The graph constructed is planar and has $4n+2$ vertices and $8n+4n/t$ edges, so we have the following lower bound for testing bipartiteness.
\begin{theorem}
	In the insertion-only model, any algorithm that can distinguish whether a planar graph of $4n+2$ vertices is bipartite or $\frac{1}{4t+2}$-far from being bipartite, must use $\Omega(n^{1-1/t})$ bits space.
\end{theorem}





\bibliographystyle{alpha}
\bibliography{streamtester}

\newpage

\appendix


\begin{center}\huge\bf Appendix \end{center}

\section{Proofs about testing $k$-edge connectivity for $k\geq 2$ from Section~\ref{sec:k-edge-connectivity}}\label{app:edge-conn}
For $k\geq 2$, Orenstein and Ron have given a characterization of graphs that are $\varepsilon$-far from being $k$-edge connected~\cite{OR11:eulerianity} (which simplifies the corresponding result in~\cite{GR02:property}). We define a subset $C$ to be \emph{$\ell$-extreme} if $|E(C,V\setminus C)|= \ell<k$ and for any $C'\subset C$, $|E(C',V\setminus C')|> \ell$. 
\begin{lemma}[Corollary 14 and Claim 16 in~\cite{OR11:eulerianity}]\label{lem:eps_far_edge_connectivity}
	If $G$ is $\varepsilon$-far from $k$-edge-connected, then there are at least $\frac{2\varepsilon m}{k}$ disjoint subsets with an edge-cut smaller than $k$. For each such a subset $C$, it contains a minimal subset $C'\subseteq C$ that is $\ell$-extreme for some $\ell<k$. 
\end{lemma}

Now we present the proof of Theorem~\ref{thm:kedgeconn}.
\begin{proof}[Proof of Theorem~\ref{thm:kedgeconn}]
	It is clear that $m\ge nk/2$ for any $k$-connected graph, and thus we can safely reject whenever $m<nk/2$. In the following, we will only consider the case that $m\ge nk/2$. 
	Our tester is then described in {\bf Algorithm} \ref{alg:test_kedgeconnectivity}.

	\begin{algorithm}
		\caption{TestKEdgeConnectivity}
		\label{alg:test_kedgeconnectivity}
		\begin{algorithmic}[1]
			\State
			Sample each vertex with probability $p:=(\varepsilon n/4k)^{-\varepsilon}$. If more than $16np$ vertices are sampled, abort and output {\bf Fail}. Let $S$ denote the set of sampled vertices.
			\State For each $v\in S$, maintain a $(k-1)$-sparse recovery sketch $\mathcal{S}_{k-1}(\edgesketch^v)$.
			\State  Maintain an AGM sketch of $G[S]$ using Theorem~\ref{thm:agm12_spanning_tree}.  
			\State {\bf Post-Processing:} 
			\State 
			Use the above sketch to recover a spanning forest $F$ of $G[S]$ using Theorem \ref{thm:agm12_spanning_tree}. 
			
			\State For each component $C\in F$, recover $\edgesketch^C$ from $\mathcal{S}_{k-1}(\edgesketch^C)$, and if it succeeds, {\bf Reject}.
			\State {\bf Accept}.
		\end{algorithmic}
	\end{algorithm}
	Note that the AGM sketch use space $\tilde{O}(|S|)=\tilde{O}(np)=\tilde{O}(k^\varepsilon n^{1-\varepsilon})$. In addition, each sampled vertex only needs to store a $k$-sparse recovery sketch, so the space complexity of the algorithm is $\tilde{O}(k)\cdot np=\tilde{O}(k^{1+\varepsilon}n^{1-\varepsilon})$. The post-processing time is near linear in the space, and the update time is $O(\poly\log n)$.
	
	For the correctness of the algorithm, we first note that if $G$ is $k$-edge connected, then $G$ will be accepted as long as there is no error happening when querying the $k$-sparse recovery sketches. This happens with probability $1-1/n$ by setting the error probability of the sketch to be $1/n^2$, and we will condition on this event. 
	
	Now if $G$ is $\varepsilon$-far from being $k$-edge connected, then from Lemma \ref{lem:eps_far_edge_connectivity}, it follows that there are at least $\frac{2\varepsilon m}{k}\ge\varepsilon n$ disjoint $\ell$-extreme subsets. Let $B_1,\cdots,B_s$ be the set of these $\ell$-extreme subsets where $s\ge\varepsilon n$. Observe that for any $\ell$-extreme subset $B$, the induced subgraph $G[B]$ is connected. This is true since otherwise, there exists a subset $B'\subset B$ satisfying $|E(B', B\setminus B')|=0$, which implies that $|E(B',V\setminus B')|\le |E(B,V\setminus B)|=\ell$, contradicting to the assumption that $B$ is $\ell$-extreme.
	
	Let $\mathcal{E}_i$ be the event that $B_i$ is entirely sampled out, and $\mathcal{F}_i$ be the event that none of the vertices in $\Gamma(B_i)$ is sampled. 
	Note that our algorithm will {\bf reject} if $\mathcal{E}_i\wedge\mathcal{F}_i$ happens for some $i$, and thus 
	our theorem will follow from the inequality that 
	\begin{eqnarray}
	\Pr\left[\bigvee_i (\mathcal{E}_i\wedge\mathcal{F}_i)\right]\geq \frac34. \label{inq:event_kconn}
	\end{eqnarray}
	
	Now we prove inequality~(\ref{inq:event_kconn}). Note that the events $(\mathcal{E}_i\wedge\mathcal{F}_i)$ are not necessarily independent across $i$ since two different $\ell$-extreme subsets may contain neighbors of each other or share neighbors. We have the following simple claim to deal with this issue.
	\begin{claim}\label{lem:independentSet}
		There exists a set $I\subset [s]$, with $|I|= s/k$, such that: 
		\begin{enumerate}
			\item\label{claim:1} $|E(B_i,B_j)|=0$ for all $i,j\in I$ and $i\neq j$, and
			\item\label{claim:2} $\sum_{i\in I}|B_i| \le \sum_{j=1}^s|B_j|/k$.
		\end{enumerate}  
	\end{claim}
	
	\begin{proof}
		We say $B_i$ and $B_j$ are neighbors, if $|E(B_i,B_j)|>0$. We iteratively construct the index set $I\subset [s]$ as follows. We start from the empty set and add one index at each step. Let $I_t$ denote the set that at the end of step $t$. In the $(t+1)$-th step, we pick the smallest set $B_j$ that is not a neighbor of $B_h$ for any $h\in I_t$. Note that since each $\ell$-extreme set has at most $k-1$ neighbors, we can always find such a set if $t<s/k$. Let $I= I_{s/k}$. Then Item~\ref{claim:1} of the claim follows by our construction. Now let $B^{(t)}$ be the set that we picked in the $t$-th step. Since each $B^{(t)}$ may intersect with at most $k$ sets, and $B^{(t)}$ is the smallest set that has no intersection with all sets picked in the first $t-1$ steps, there must exist a partition of $[s]$ into $s/k$ sets $\{P_1,P_2,\cdots, P_{s/k}\}$, such that for any $t\leq s/k$ and $j\in P_t$, $|B_j|\ge |B^{(t)}|$. This proves Item~\ref{claim:2} of the claim.
	\end{proof}
	
	Now we give a lower bound for $\Pr[\bigvee_{i\in I}\mathcal{E}_i]$. Let $p_i=p^{|B_i|}$ be the probability that all vertices in $B_i$ are sampled. Using the fact $1-x\le e^{-x}$ for all $x$ and the AM-GM inequality, we have
	\begin{eqnarray*}
		\ln \prod_{i\in I}(1-p_i)\le {-\sum_{i\in I}p_i}\le {- |I|\cdot(\prod_{i\in I}p_i )^{1/|I|} } = {-|I|\cdot p^{\sum_{i\in I}|B_i| /|I|}} & =& -(s/k)\cdot p^{\sum_{i} |B_i|/s} \\
		&\le& -\frac{\varepsilon n}{k}\cdot p^{1/\varepsilon}.
	\end{eqnarray*}
	Thus we have 
	$$\Pr\left[\bigvee_{i\in I}\mathcal{E}_i\right]=1-(\prod_{i\in I}(1-p_i))\ge 1-e^{-\frac{\varepsilon n}{k}\cdot p^{1/\varepsilon}}\ge 15/16,$$
	since we set $p=(\varepsilon n/4k)^{-\varepsilon}$.
	
	
	Now by the property of $I$ as guaranteed in Claim~\ref{lem:independentSet}, it follows that $\mathcal{F}_j$ and $\mathcal{E}_i$ are independent for all $i,j\in I$. Hence, conditioned on the event $\bigvee_{i\in I}\mathcal{E}_i$, the probability of $\bigvee_{i\in I} (\mathcal{E}_i\wedge\mathcal{F}_i)$ happening is
	$$\Pr\left[\bigvee_{i\in I} (\mathcal{E}_i\wedge\mathcal{F}_i) \Bigm\vert \bigvee_{i\in I}\mathcal{E}_i \right]\ge\min_{j\in I} \Pr[\mathcal{F}_j]=\min_{j\in I} (1-p)^{|\Gamma(B_j)|}\ge (1-p)^k\ge e^{-pk-p^2k}\ge 0.8,$$
	where in the penultimate inequality, we used the basic inequality that $1-x\ge e^{-x-x^2}$ for $x\le 0.5$; the last inequality holds for $k\le 0.1/p$ or equivalent $k\le O(n^{\varepsilon/(1+\varepsilon)})$.
	
	
	Finally, we have 
	\begin{eqnarray*}
		\Pr\left[\bigvee_i (\mathcal{E}_i\wedge\mathcal{F}_i)\right]\ge \Pr\left[\bigvee_{i\in I} (\mathcal{E}_i\wedge\mathcal{F}_i)\right] &=&\Pr\left[\left(\bigvee_{i\in I}(\mathcal{E}_i\wedge\mathcal{F}_i)\right)\bigwedge \left(\bigvee_{i\in I}\mathcal{E}_i\right) \right]\\
		&=&\Pr\left[\bigvee_{i\in I}\mathcal{E}_i\right]\cdot\Pr\left[\bigvee_{i\in I} (\mathcal{E}_i\wedge\mathcal{F}_i) \Bigm\vert \bigvee_{i\in I}\mathcal{E}_i \right] \\
		&\ge&\frac{15}{16}\cdot\frac{4}{5}\ge 3/4. 
	\end{eqnarray*}
\end{proof}
We remark that the problem can still be solved in space $\tilde{O}(kn^{1-\varepsilon})$ for larger $k$ by testing the neighborhood of all subsets of size smaller than $1/\varepsilon$ in $S$, however the post-processing time will be $\tilde{O}(kn^{O(1/\varepsilon)})$. Also, $k\le O(n^{\varepsilon})$ is the most interesting case for us, since we are mostly interested in $o(n)$ space algorithms.

\section{Testing other graph properties}\label{app:otherProperties}
\subsection{$k$-vertex connectivity}\label{app:kvertex}
A graph is $k$-vertex connected if the minimum vertex cut of the graph has size at least $k$, i.e. it remains connected whenever fewer than $k$ vertices are removed. The following lemma on the structure of graphs that are $\varepsilon$-far from being $k$-vertex connected can be directly deduced from Corollary 19 in~\cite{OR11:eulerianity}.
\begin{lemma}\label{lem:eps_far_vertex_conn}
	If the graph is $\varepsilon$-far from $k$-vertex connected, then there exists at least $\frac{\varepsilon m}{2k}$ subsets $C$ of size at most $\frac{2k n}{\varepsilon m}$ such that $G[C]$ is connected and $\Gamma(C)< k$.
\end{lemma}
\begin{proof}[Proof sketch]
	In Corollary 19 in~\cite{OR11:eulerianity}, it is proven that for any directed graph $G$ that is $\varepsilon$-from $k$-vertex connected, then there exists at least $\frac{\varepsilon m}{2k}$ subsets $C$ of size at most $\frac{2kn}{\varepsilon m}$, and either $\Gamma^+(C)< k$ or $\Gamma^-(C) < k$. 
	
	On the other hand, in Section 5.3 in~\cite{OR11:eulerianity}, it is proven that if $G$ is $\varepsilon$-far from $k$-vertex connected, then the corresponding directed graph $G'$ that is obtained by turning each undirected edge $(u,v)$ into directed edges $\agbracket{u,v}$ and $\agbracket{v,u}$ is $\varepsilon$-far from being $k$-vertex connected. Therefore, there exists at least $\frac{\varepsilon m}{2k}$ subsets $C$ in $G'$ of size at most $\frac{2k n}{\varepsilon m}$, and either $\Gamma_{G'}^+(C)< k$ or $\Gamma_{G'}^-(C) < k$. This directly implies that the corresponding set $C$ in $G$ satisfies that $\Gamma_{G}(C)<k$. Finally, if $G[C]$ is not connected, then we can replace $C$ by one maximal subset $C'\subset C$ such that $G[C']$ is connected. Note that $\Gamma_{G}(C')\leq \Gamma_G(C)<k$. This completes the proof of the lemma.
\end{proof}

\begin{theorem}~\label{thm:kvertexconn}
	Let $k\le O(n^{\varepsilon/(4+\varepsilon)})$. There exists a single-pass dynamic streaming tester for $k$-vertex connectivity with post-processing time and space complexity $\tilde{O}(\frac{k^{1+\varepsilon/4}}{\varepsilon}\cdot n^{1-\varepsilon/4})$.
\end{theorem}
\begin{proof}[Proof sketch]
	We can also simply consider the case that $m\geq nk/2$, since otherwise the graph cannot be $k$-vertex connected and we can directly reject.
	Our approach for testing $k$-connectivity is similar to testing $k$-edge connectivity.
	The difference here is that now we cannot use the $(k-1)$-sparse recovery sketch for the vector $\edgesketch^v$. Instead, for each vertex $v\in S$, we will maintain an exact $k'$-sparse recovery sketch of the vector $\vertexsketch^v$ (defined in section \ref{sec:notation}), $\mathcal{S}_{k'}(\vertexsketch^v)$, for $k'=\frac{4}{\varepsilon }+k$. Then for each detected connected component $C$ of size smaller than $4/\varepsilon$ in $G[S]$ (by AGM sketch), recover $\vertexsketch^C:=\sum_{v\in C}\vertexsketch^v$ from the sketch $\mathcal{S}_{k'}(\vertexsketch^C)=\sum_{v\in C}\mathcal{S}_{k'}(\vertexsketch^v)$. If it succeeds, we get the set $C\bigcup \Gamma(C)$, and since we know $C$, we get $\Gamma(C)$. If $|\Gamma(C)|<k$, we {\bf reject}. For any $k$-vertex connected graph, the tester will never {\bf reject} if all the sparse recover sketches return correctly, which happens with high probability. On the other hand, if $G$ is $\varepsilon$-far from $k$-vertex connected, by similar analysis as in $k$-edge connectivity together with Lemma \ref{lem:eps_far_vertex_conn}, we know that with high probability, there is a subset $C\subseteq S$ such that $G[C]$ is a connected component in $G[S]$, $|\Gamma(C)|<k$ and $|C|\le 4/\varepsilon$, and conditioned on this the algorithm will successfully recover $\Gamma(C)$, and reject with high probability. Here to make the analysis work, we have to set the sampling probability $p:=(\varepsilon n/16k)^{-\varepsilon/4}$, so the space used is $\tilde{O}(k'\cdot k^{\varepsilon/4}\cdot n^{1-\varepsilon/4})=\tilde{O}(\frac{k^{1+\varepsilon/4}}{\varepsilon}\cdot n^{1-\varepsilon/4})$. Since the analysis is almost the same as $k$-edge connectivity, we omit the details here. 
\end{proof}

\subsection{Testing Eulerianity}\label{app:testingEuler} 
Note that the algorithm for connectivity testing can be directly used to testing Eulerianity. A graph $G$ is Eulerian if there is a path in the graph that traverses each edge exactly once, which is equivalent to that $G$ is connected and the degrees of all vertices are even or exactly two vertices have odd degrees. Note that if graph $G$ is $\varepsilon$-far from being Eulerian then either $G$ has $\Omega(\varepsilon n)$ connected components (i.e. far from being connected) or has $\Omega(\varepsilon n)$ vertices of odd degree (cf., \cite{GR02:property,PR02:diameter}). Then one can test Eulerianity by first invoking the previous algorithm on testing connectivity, and then sample $O(1/\varepsilon)$ vertices and check if some sampled vertex has odd degree. The post-processing time and space complexity of the final algorithm are $\tilde{O}(n^{1-c\cdot\varepsilon})$ for some universal constant $c$. 

\end{document}